\documentclass[11pt,a4paper]{article}
\usepackage{epsf,epsfig,amsfonts,amsgen,amsmath,amstext,amsbsy,amsopn,amsthm}
\usepackage{amssymb}
\usepackage{cite}
\usepackage{enumitem}
\newtheorem{theorem}{Theorem}

\newtheorem{lemma}[theorem]{Lemma}
\newtheorem{cor}{Corollary}

\theoremstyle{definition}
\newtheorem{definition}{Definition}
\newtheorem{claim}{Claim}

\newtheorem{remark}[claim]{Remark}

\newtheorem{example}{Example}
\begin{document}
\title{\bf {\Large On the largest  minimum distances of $[n,6]$  LCD codes}}
\date{}
\author{Yang Liu$^{a,b,\ast}$, Ruihu Li$^{b}$\\
 $^{a}$ Air Defense and antimissile school, Air Force Engineering
 University,\\ Xi'an, Shaanxi 710051,  China.\\
 $^{b}$ Department of Basic Sciences, Air Force Engineering University, \\ Xi'an, Shaanxi 710051,
 China.
(email:$^{\ast}$liu\_yang10@163.com)\\
} \maketitle

\begin{abstract}
Linear complementary dual (LCD) codes can be used to against
side-channel attacks and fault noninvasive attacks. Let $d_{a}(n,6)$
and $d_{l}(n,6)$ be the minimum weights of all binary optimal linear
codes and LCD codes with length $n$ and dimension 6, respectively.
In this article,  we aim to obtain  the values of $d_{l}(n,6)$ for
$n\geq 51$ by investigating the nonexistence and constructions of
LCD codes with given parameters. Suppose that $s \ge 0$ and $0\leq
t\leq 62$ are two integers and $n=63s+t$. Using   the theories of
defining vectors, generalized anti-codes, reduced codes and nested
codes, we exactly determine $d_{l}(n,6)$ for $t \notin
\{21,22,25,26,33,34,37,38,45,46\}$, while we show that $
d_{l}(n,6)\in$$\{d_{a}(n,6)$  $-1,d_{a}(n,6)\}$ for $t\in
\{21,22,26,34,37,38,46\}$ and $ d_{l}(n,6)\in$$ \{d_{a}(n,6)-2,$
$d_{a}(n,6)-1\}$ for $t\in\{25,33,45\}$.
\medskip

\end{abstract}

\section{\label{sec:level1} Introduction\protect}

\label{sec1} Let $F_{2}^{n}$ be the $n$-dimensional row vector space
over binary field $F_{2}$. A binary linear code $\mathcal{C}=[n,k]$
is a $k$-dimensional subspace of $F_{2}^{n}$. The weight $w(x)$ of a
vector $x\in F_{2}^{n}$ is the number of its nonzero coordinates. If
the minimum weight of nonzero  vectors in $\mathcal{C}$$=[n,k]$ is
$d$, then $d$ is called the minimum distance of $\mathcal{C}$ and we
denote $\mathcal{C}=[n,k,d]$.
 The dual code $\mathcal{C}^{\perp}$ of $\mathcal{C}$ is defined as
$\mathcal{C}^{\perp}$= $\{x\in F_{2}^{n} \mid x\cdot y=xy^{T}=0
 \hbox{ for all}\  y \in  \mathcal{C} \}$.
  The hull (or
radical code) of $\mathcal{C}$ is defined as
$Hu(\mathcal{C})=\mathcal{C}^{\perp}\cap\mathcal{C}$ and
$h(\mathcal{C})=$dim$Hu(\mathcal{C})$ is the hull dimension of
$\mathcal{C}$ \cite{Assmus}.  If a generator matrix of a linear code
is denoted by $G$, then it is easy to know $h(\mathcal{C})=k-$
rank$(G^{T}G)$. A code $\mathcal{C}$ is self-orthogonal (SO) if
 $h(\mathcal{C})=k$. And if  $h(\mathcal{C})=0$,
$\mathcal{C}$ is a linear complementary dual (LCD) code.

LCD cyclic codes were referred to as reversible cyclic codes  by
Massey \cite{massey1964} and could provide an optimal linear coding
solution for the two user binary adder channel. He also  pointed out
that asymptotically good LCD codes exist \cite{massey1992}. Soon
afterwards Yang et al. provided the  condition for a cyclic code
having a complementary dual \cite{Yang1994}. In Ref.
\cite{Sendrirer2004} Sendrier verified that LCD codes can meet the
Gilbert-Varshamov bound. Carlet et al. showed that LCD codes can be
used  to fight against fault noninvasive attacks and  side-channel
attacks \cite{carlet2015}. It shed a new light on LCD codes and
posed more attention to the construction of LCD codes with greatest
possible minimum distance, which can improve the resistance against
those two attacks. Since then, much work aimed at determining the
upper bound of the minimum distances of LCD codes and constructing
LCD codes with  best parameters, see
\cite{carlet2018,Grassl,Galvez2018,Harada2019,Fu2019,Araya2020,Bouyuklieva2021,Wang2023,liuar,Shi22,Araya2021char,Araya2021}.

In this paper,  we use $d_{a}(n,k)$ to denote the distance of
optimal linear code for given $n, k$, and $d_{l}(n,k)$ to denote the
largest distance of LCD codes for given $n, k$. A linear $[n, k]$
code $\mathcal{C}$ is optimal if $\mathcal{C}$ has the greatest
minimum weight $d_{a}(n,k)$ among all linear $[n, k]$ codes. And a
linear $[n,k,d]$ code with $d=d_{a}(n,k)-1$ is called a near optimal
code.
 If an LCD $[n, k]$ code with the largest minimum weight
$d_{l}(n,k)$ among all LCD $[n, k]$ codes, then it is an optimal LCD
code. And an LCD $[n,k,d]$ code with $d=d_{l}(n,k)-1$ is called a
near optimal LCD code.

Ref. \cite{carlet2018} showed that any code over $F_{q}$ is
equivalent to an LCD code for $q \geq 4$. This  suggests us pay more
attention on LCD codes over $F_{q}$ for $q=2, 3$. Here we only
consider binary LCD codes.
 In latest years,
constructions of optimal  LCD  codes with short lengths or low
dimensions are discussed, and the lower and upper bounds for
$d_{l}(n,k)$ have been established  in
\cite{Grassl,Galvez2018,Harada2019,Fu2019,Araya2020,Bouyuklieva2021,Wang2023,liuar,Shi22,Araya2021char,Araya2021}.
For $n\leq 24$ and $1\leq k\leq n$, all  $d_{l}(n,k)$ were
determined.  For $k\leq n\leq 40$, most of $d_{l}(n,k)$ were
determined.

For $k\leq 5$, all $d_{l}(n,k)$ were obtained in
\cite{Harada2019,Fu2019,Araya2020,liuar,Araya2021char,Araya2021}.
 As for larger $n$ and
$k\geq 6$, only a few results of $d_{l}(n,k)$ are known, see
\cite{Galvez2018,Harada2019,Fu2019,Araya2020,Fu2019,Bouyuklieva2021,Wang2023,Li2022,Li2023}.
When $k=6$, $d_{l}(n,6)$ was given for $1\leq n \leq 12$ in
\cite{Galvez2018}, for $13\leq n \leq 16$ in \cite{Harada2019}, for
$17\leq n \leq 24$ in \cite{Fu2019,Araya2020}, for $25\leq n \leq
40$ in \cite{Fu2019,Bouyuklieva2021}, for $41\leq n \leq 50$ in
\cite{Wang2023}.

Recently,  Li {\it et al.}  introduced a new concept ``generalized
anti-code" to verify that 11 classes of binary optimal codes
$[2^{k-1}s + 2^{k}-a, k]$ are not LCD codes for given $k$ and $a$ in
\cite{Li2023} (Lemma \ref{lilem} in this article).  The new approach
is useful to study optimal LCD codes with higher dimensions. For
example, they showed that some optimal linear codes are not LCD for
$k=6$ and $52\leq n\leq 62$. However, they hadn't accurately
provided the upper bound of the minimum distances of optimal linear
codes.

Motivated by the above results, the objective of this paper is to
investigate  the values of $d_{l}(n,6)$ for all lengths $n\geq 51$.
This paper is organized as follows. In Section \ref{sec2}, some
definitions, notations and basic results  about  defining vectors,
generalized anti-codes and reduced codes will be  are given. In
Sections \ref{sec3} and \ref{sec4},  we will determine the minimum
distances of optimal $[n=63s+t, 6]$ LCD codes by the following two
steps:

Step1:  We  infer that many (near) optimal linear codes with given
lengths are not LCD, which will be showed in Section \ref{sec3}.

Step2:  We construct some (near) optimal  LCD codes with parameters
$[63s+t,6,d_{l}]$  from some nested code chains, which will be
exhibited in Section \ref{sec4}.

In final section  we will  provide a table to conclude the above
results and give  the minimum distances of optimal $[n=63s+t, 6]$
LCD codes.

\section{Preliminaries}\label{sec2}
In this section,  some concepts, lemmas and   notations about
defining vectors, generalized anti-codes and reduced codes will be
introduced. For clarity, they are divided into the following four
subsections.

\subsection{Basic knowledge}

\begin{lemma}(Griesmer Bound)\cite{bro0}\quad
The length, dimension and minimum distance for all linear $[n, k,
d]$ codes over $F_{q}$ achieve the following relation:
$$n\geq \sum_{i=0}^{k-1}\lceil\frac{d}{q^{i}}\rceil$$
\end{lemma}

It should be noted that all linear codes meeting the Griesmer Bound
are optimal codes while not all optimal codes can achieve  the
Griesmer Bound.

Two binary codes $\mathcal{C}$ and $\mathcal{C}'$ are equivalent  if
one can be obtained from the other by permuting the coordinates, we
denote them as  $\mathcal{C}$ $ \cong$ $\mathcal{C}'$.
  If two matrices  $G_{1}$ and $G_{2}$ generate equivalent codes,  we denote them as $G_{1}\cong$ $G_{2}$.
Two $k\times n$ matrices $G_{1}\cong$ $G_{2}$ if and only if there
are an invertible matrix $P$ and a permutation $\Pi$ on
$\{1,2,\cdots,n\}$ such that $G_{2}=PG_{1}\Pi$ \cite{Huffman}.

\subsection{Definitions about the defining vector}

 We  use
$\bf{1_{n}}$=$(1,1,\cdots,1)_{1\times n}$ and
$\bf{0_{n}}$=$(0,0,...,0)_{1\times n}$ to denote the all-one vector
and the zero vector of length $n$, respectively. Let $J_{k}$ be the
$(2^{k}-1)\times (2^{k}-1)$ all-one matrix. And use
$iG=(G,G,\cdots,G)$ to denote the juxtaposition  of $i$ copies of
$G$ for given matrix $G$. We only consider linear codes  without
zero coordinates and matrices without zero columns.

Let $\alpha_{i}$ be the binary column vector representation of $i$
for $1\leq i\leq N=2^{k}-1$,  i.e.,
$\alpha_{1}$=$(1,0,\cdots,0)^{T}$,
$\alpha_{2}$=$(0,1,\cdots,0)^{T}$, $\cdots$,
$\alpha_{N}$=$(1,1,\cdots,1)^{T}$. Then the matrix
$\mathbf{S}_{k}=(\alpha_{1},\alpha_{2},\cdots,\alpha_{N})$ generates
the $k$-dimensional simplex code
$\mathcal{S}_{k}=[2^{k}-1,k,2^{k-1}] $. Recall the matrix
$\mathbf{S}_{k}$ can be constructed inductively \cite{Zuo2011,Zuo2}.
Let
\begin{center}
$\scriptsize{\mathbf{S}_{2}=\left( \begin{array}{lllllll}
101\\
011\\
\end{array} \right)}$, $\scriptsize{\mathbf{S}_{3} =\left(
\begin{array}{ccccccccccccc}
\mathbf{S}_{2}& \mathbf{0}^{T}_{2}&\mathbf{S}_{2}\\
 \mathbf{0}_{3}&1&\mathbf{1}_{3}\\
\end{array} \right)}$
\end{center}
and recursively, we have
\begin{center}
 $\scriptsize{\mathbf{S}_{k+1}=\left(
\begin{array}{ccccccccccccc}
\mathbf{S}_{k}& \mathbf{0}^{T}_{k}&\mathbf{S}_{k}\\
 \mathbf{0}_{2^{k}-1}&1&\mathbf{1}_{2^{k}-1}\\
\end{array} \right)}$.
\end{center}

 Let $G$ be a  matrix
of size $k\times n$. If the columns of  $G$ contain $l_{i}$ copies
of $\alpha_{i}$ for $1\leq i\leq N$, we denote $G$ as
$G=(l_{1}\alpha_{1},l_{2}\alpha_{2},\cdots,l_{N}\alpha_{N})$ for
short, and call $L=(l_{1},l_{2},\cdots,l_{N})$  the  {\bf defining
vector} of $G$ \cite{Li2008}. Let $l_{max}$ and $l_{min}$ be the
largest and the smallest values of all $l_{i}$ for $1\leq i\leq
N=2^{k}-1$, respectively.

Suppose $L=(l_{1},l_{2},\cdots,l_{N})$, let $l_{j_{l}}$ ($1\leq
l\leq t$) be different coordinates of $L$ with $l_{j_{1}}<
l_{j_{2}}<\cdots<l_{j_{t}}$. Defining $m_{l}$ as the number of equal
$l_{j_{l}}$, we say $L$ is of {\it type}
$$]](l_{j_{1}})_{m_{1}}\mid (l_{j_{2}})_{m_{2}}\mid
\cdots\mid(l_{j_{t}})_{m_{t}}]].$$

 If  $\mathcal{C}=[n,k]$ has a generator matrix
 $G=(l_{1}\alpha_{1},\cdots,l_{N}\alpha_{N})$,
the minimum distance $d$ of $\mathcal{C}$ and its codewords weight
can be determined by its defining vector $L$ and some special
matrices $P_{k}$ and $Q_{k}$ from simplex codes.

Let $P_{2}$ be a $(2^{2}-1)\times (2^{2}-1)$ matrix whose rows are
the non-zero codewords of $\mathcal{S}_{2} $. Construct

$$P_{2}=\left( \begin{array}{ccccccc}
101\\
011\\
110\\
\end{array} \right),
P_{3}= \left(
\begin{array}{ccccccc}
P_{2}&0&P_{2}\\
\bf{0_{3}}&1& \bf{1_{3}}\\
P_{2}&\bf{1^{T}_{3}}&(J_{2}-P_{2}).\\
\end{array}
\right)$$

 Then the seven rows of $P_{3}$  are just  the seven nonzero vectors
of the simplex code $\mathcal{S}_{3}$  $=[7, 3, 4] $.
 For $k\geq 3$, using the recursive
method, let $P_{k}$ be a $(2^{k}-1)\times (2^{k}-1) $ matrix whose
rows are just the $2^{k}-1 $ nonzero codewords of $k$-dimensional
binary simplex code. Then the matrix formed by nonzero codewords of
$k+1$-dimensional simplex code is as follows:\\
$$P_{k+1}= \left(
\begin{array}{ccccccc}
P_{k}&\bf{0^{T}_{2^{k}-1}}&P_{k}\\
\bf{0_{2^{k}-1}}&1& \bf{1_{2^{k}-1}}\\
P_{k}&\bf{1^{T}_{2^{k}-1}}&(J_{k}-P_{k})\\
\end{array}
\right).$$\\

Let $Q_{k}=J_{k}-P_{k}$, Each row of $Q_{k}$ has $2^{k-1}-1$ ones
and $2^{k-1}$ zeros. According to Refs. \cite{Zuo2011,Zuo2,Li2008},
the matrix $P_{k}$ is invertible over the rational field and
$$P^{-1}_{k}=\frac{1}{2^{k-1}}(J_{k}-2Q_{k})=\frac{1}{2^{k-1}}(2P_{k}-J_{k}).$$

 For an $[n,k,d]$ code with defining vector
$L=(l_{1},l_{2},\cdots,l_{N})$, the weights of its $2^{k}-1$ codes
form a vector $W=(w_{1},w_{2},\cdots,w_{N})$, called  its weight
vector. Set
$W=(w_{1},w_{2},\cdots,w_{N})=d\bf{1_{2^{k}-1}}+\Lambda$, where
$\Lambda=(\lambda_{1},\lambda_{2},\cdots,\lambda_{N})$ with
$\lambda_{i}=w_{i}-d\geq 0$ and at least one $\lambda_{i}=0$. Denote
$\sigma=\lambda_{1}+\lambda_{2}+\cdots+\lambda_{N}$. Then the weight
vector $W$ and defining vector $L$ have the following relations:
\begin{eqnarray*}
W^{T}&=&P_{k}L^{T}, \sigma=2^{k-1}n-d(2^{k}-1)\\
  L^{T} &=&
  P^{-1}_{k}W^{T}=\frac{1}{2^{k-1}}(J_{k}-2Q_{k})(d\bf{1_{2^{k}-1}}+\Lambda)^{T}\\
  &=&\frac{1}{2^{k-1}}((d+\sigma){\bf 1}_{2^{k}-1}^{T}-2Q_{k}\Lambda^{T})
\end{eqnarray*}

For an $[n,k,d]$ code with defining vector
$L=(l_1,l_2,\cdots,l_{2^{k}-1})$, one can infer from the above
equations that
$$\lfloor \frac{1}{2^{k-1}}(d+\sigma)\rfloor\leq l_{i} \leq \lceil \frac{1}{2^{k-1}}(d-\sigma)\rceil,$$
which will be frequently used in the process of verifying the
 below conclusions.

\subsection{Definitions about the generalized anti-code}

Binary anti-codes were first introduced  by Farrell in
\cite{Far1970}. To verify  the nonexistence of binary optimal LCD
codes $[2^{k-1}s + 2^{k}-a, k]$  for given $k$ and $a$,  Li {\it et
al.} proposed some concepts about the generalized anti-code in
\cite{Li2023}.

\begin{definition}\cite{Li2023} Suppose  $\mathcal{C}=[n,k]$ has a generator matrix
$G=(l_{1}\alpha_{1},\cdots,l_{N}\alpha_{N})$ and its defining vector
is $L=(l_{1},\cdots,l_{N})$.
 If $l_{max}=a$, then $G$ is a sub-matrix of $aS_{k}$,
$L^{c}=((a-l_{1}),\cdots,(a-l_{N}))$ is called the anti-vector of
$L=(l_{1},\cdots,l_{N})$ and it is the defining vector of $G^{c}=((
 a-l_{1})\alpha_{1},\cdots,(a-l_{N})\alpha_{N})$.
 $G^{c}$  is also a sub-matrix of $aS_{k}$ and called an anti-matrix of $G$.
A linear code generated by $G^{c}$ is called as a {\bf generalized
anti-code}. If the largest weight of these $2^{k}$-vectors is
$\delta$, this generalized  anticode is denoted as
$\mathcal{C}^{a}=(m,2^{k},\{\delta\})$.
\end{definition}


Parameters of $\mathcal{C}=[n,k]$ with generator matrix $G$ can be
derived from its  generalized anti-code     as follows.

\begin{lemma}\cite{Li2023} \label{li2023lemma}   Let
$G=(l_{1}\alpha_{1},l_{2}\alpha_{1},\cdots,l_{N}\alpha_{N})$ be a
generator matrix of $\mathcal{C}=[n,k,d]$, where  $N=2^{k}-1$. Set
$a=max_{1\leq i\leq N}\{l_{i}\}$, $n<aN$,  $m=aN-n$ and
$l^{c}_{i}=a-l_{i}$. If
$G^{c}=(l^{c}_{1}\alpha_{1},l^{c}_{2}\alpha_{1},\cdots,l^{c}_{N}\alpha_{N})$
generates a  generalized anti-code
$\mathcal{C}^{a}=(m,2^{k},\{\delta\})$, then $d=a2^{k-1}-\delta$ and
$GG^{T}=G^{c}(G^{c})^{T}$ for $k\ge 3$.
\end{lemma}

 Combining Ref. \cite{massey1992} with Ref. \cite{Li2023}, equivalent conditions for an LCD code are given as
follows.

\begin{lemma} \cite{massey1992,Li2023} \label{li3} Let $k\geq 3$, $\mathcal{C}=[n,k]$  be a linear
code with  a generator matrix $G$  and  a parity check matrix $H$.
If the anti-matrix of  $G$ is $G^{c}$, then the following properties
are equivalent:

 \noindent (1)   $\mathcal{C}$ is LCD;

 \noindent (2)  the matrix $HH^{T}$ is invertible;

 \noindent (3)  the matrix $GG^{T}$ is invertible;

\noindent  (4)  the matrix $G^{c}(G^{c})^{T}$ is invertible.
\end{lemma}

The following two examples can make these previous concepts
understandable.

\begin{example} If a linear code has a defining vector $L_{1}$ and
 $L_{1}=(s+1,s-1,s,s,s+1,s-1,s+1)$ is of type $]](s-1)_{2}\mid
(s)_{2}\mid(s+1)_{3}]]$. One can know $l_{max}=s+1$ is the largest
value of all $l_{i}$ for $1\leq i\leq N=2^{3}-1=7$, then $G$ is a
sub-matrix of $(s+1)S_{3}$ and
$L_{1}^{c}=((s+1-l_{1}),\cdots,(s+1-l_{7}))=(0,2,1,1,0,2,0)$ is  the
anti-vector of $L=(l_{1},\cdots,l_{7})$. Here $L_{1}^{c}$ is the
defining vector of
$G^{c}=(3\alpha_{1},\alpha_{2},\cdots,\alpha_{7})$ and of type
$]](2)_{2}\mid (1)_{2}\mid(0)_{3}]].$
\end{example}

 \begin{example} If a linear code has a defining vector $L_{2}=(3,1,1,3,1,3,1)$, then one can know  $L_{2}$ is
of type $]](1)_{4}\mid (3)_{3}]]$ and it is easy to derive that
$L_{2}^{c}$ is of type $]](2)_{4}\mid(0)_{3}]]$ similar to the above
example. From the form of $L_{2}^{c}$,  one can  infer that a code
with defining vector $L_{2}$ is an SO code owing to rank
$(GG^{T})$=rank$(G^{c}(G^{c})^{T})$=0. Naturally, it is not LCD by
lemma \ref{li3}.

\end{example}

\subsection{Some preliminary results}

The hull $h$ and $d_l$  of a linear code with parameters $[n,k]$ can
be estimated from extended codes or codes with lower dimensions,
which are given as follows.

\begin{definition}\cite{Li2022} Let $G$, $G_{1}$ be  generator matrices of
$\mathcal{C}=[n, k,d]$ and $\mathcal{C}_{1}=[n-m, k-1,\geq d]$,
respectively. If
$$G=\left(\begin{array}{ccccc}
                   \mathbf{1}_{m}  & u \\
                      \mathbf{0}_{m} &G_{1}\\
                    \end{array}
                 \right),$$ \\
then $\mathcal{C}_{1}$ is called a {\bf reduced code} of
$\mathcal{C}$.
\end{definition}

\begin{lemma}\cite{Li2022}\label{li} If $\mathcal{C}_{1}$ is a reduced code of
$\mathcal{C}=[n, k,d]$ and $h(\mathcal{C}_{1})=r\geq 2$, then
$h(\mathcal{C})\geq r-1$ and $\mathcal{C}$ is not an LCD code.
\end{lemma}


\begin{lemma} \cite{Bouyuklieva2021}\label{Bou} If $k$ is even and
$d_{l}(n,k)$ is odd, then $d_{l}(n+1, k)\geq d_{l}(n, k)+1$.
\end{lemma}

\begin{lemma}\cite{Mac60} \label{li2}Let $s\geq 0$, $k\geq 6$, $1\leq m\leq k-1$, then
the  code $\mathcal{C}=$$\mathcal{MD}$$_{s}(k,m)$$=[sN+2^{k}-2^{m},
k,s2^{k-1}+2^{k-1}-2^{m-1}]$ has $h(\mathcal{C})=k$ for $m=0$,
$h(\mathcal{C})=k-1$ for $m=1$ and $h(\mathcal{C})=k-2\geq 4$ for
$m\ge 2$, hence $\mathcal{C}$ is not an LCD code.
\end{lemma}

\section{ On the nonexistence of $[n,6,d_{n}]$  linear  LCD codes}
\label{sec3}

In this section, let  $s\geq 1$ be an integer and we consider the
nonexistence of LCD codes $[n,6,d_{n}]$. For clarity,  some known
results about the Griesmer bounds $d_g$, the minimum distances $d_a$
of all optimal linear codes and  the minimum distances $d_l$ of
optimal LCD codes are listed in Table 1.

\begin{table}[h!]\label{t1}
    \begin{center}
        \caption{Known results of $[n, 6, d_{t}]$  codes in Refs.
        \cite{Grassl,Galvez2018,Harada2019,Fu2019,Araya2020,Bouyuklieva2021,Wang2023} for $6\le n\le 50$}
        $\begin{tabular} {lllllllllllllllllllllllllllll}
            \hline \hline
            $n$ &  6& 7& 8& 9& 10& 11& 12& 13&14& 15& 16& 17&  18& 19& 20\\
            \hline
            $d_g$& 1& 2& 2& 3& 4&   4&  4&  5& 6& 6 &7&8& 8&  8& 8\\
            $d_a$& 1& 2& 2& 2& 3&   4&  4&  4& 5& 6 &6&7& 8&  8& 8\\
            $d_l$&1& 2& 2& 2& 3&   4&  4&  4& 5& 6&6&6& 7&  8& 8 \\
            \hline \hline
            $n$& 21& 22& 23& 24& 25& 26& 27& 28&29& 30& 31& 32&  33& 34& 35\\
            \hline
            $d_g$&9&    10&  10&11& 12& 12& 12& 13&  14&14& 15&16&16& 16& 16 \\
            $d_a$&8&    9&  10&10& 11& 12& 12& 12&  13&14& 15&16&16& 16& 16 \\
            $d_l$&8&    9&  10&10& 10& 11& 12& 12&  12&13 & 14&14&14& 15& 16\\
            \hline \hline
            $n$& 36& 37& 38& 39&40& 41&42& 43& 44& 45& 46& 47&  48& 49& 50\\
            \hline
            $d_g$& 16&17&18& 18& 19& 20& 20& 20& 21&  22& 22& 23& 24& 24& 24\\
            $d_a$&16&17& 18& 18& 18& 19& 20&20& 21&  22& 22& 23& 24& 24& 24\\
            $d_l$&16&16& 17& 18& 18& 19& 20&20& 20&  20& 21& 22& 22& 23& 24\\
            \hline \hline
        \end{tabular}$

$d_{t }$ is used to denote $d_g$, $d_a$ or $d_l$ in different
lines.\qquad \qquad \qquad \qquad \qquad \qquad \qquad
    \end{center}
\end{table}

From Table 1 and the code tables given  by Grassl in \cite{Grassl},
one can easily know the following items:

(1) For $n=6,7,8,11,12,15,18,19,20,23,26,27$, $30\leq n \leq
63(n\neq 40,41)$, the optimal linear $[n,6,d_a]$ codes saturate the
Griesmer bound. If $n\geq 64(s\geq 1)$, one can  infer that all
optimal $[63s+t,6,32s+d_{a}(t)]$ codes also saturate the Griesmer
bound.

(2) The greatest minimum distance $d_l$ of optimal LCD codes can not
meet the corresponding $d_a$ for
$n=18,26,30,31,32,33,34,37,38,44-49$.

(3) For $n=9,10,13,14,16,17,21,22,24,25,28,29, 40,41$,  the optimal
linear codes $[n,6,d_a]$ can not saturate the Griesmer bound.

  Thus, one should only pay special attention to LCD codes of lengths
$n=63s+t \geq 51$ for
$t=0,1,2,3,4,5,9,10,13,14,16,17,18,21,22,24,25,26,28-34,
37,38,40,41,44-49,51-62$. In the sequel, we will investigate LCD
codes with length $n=63s+t \geq 51$ for given $s,t$.

\begin{theorem} There is no  $[63s,6,32s]$, $[63s,6,32s-1]$, $[63s+1,6,32s]$,
$[63s+1,6,32s-1]$ or $[63s+2,6,32s]$ LCD code.

\end{theorem}

\begin{proof} Obviously, a $[63s,6,32s]$ code is the juxtaposition of $s$ simplex codes with dimension 6 and then naturally SO.

 If $\mathcal {C}$ is a $[63s+1,6,32s]$ linear code, then
$\sigma=32$ and its defining vector $L$ satisfies $l_{max}=s+1$. It
follows that a $[63s+1,6,32s]$ code has a reduced code
$[62s,5,32s]=[31\times 2s,5,16\times 2s]$,  which is SO. Thus, one
can deduce a code with parameters $[63s+1,6,32s]$ has $h\geq 4$ and
a $[63s+1,6,32s]$ code is not LCD  by Lemma \ref{li}. Further, there
exists no $[63s,6,32s-1]$ LCD code considering  that a
$[63s,6,32s-1]$ code can be extended to  a $[63s+1,6,32s]$ LCD code
by Lemma \ref{Bou}.

 A $[63s+2,6,32s]$ code has $\sigma=64$ and then its defining vector  $L$ satisfies $s-1 \le l_{i}\le s+2$.

If $l_{max}= s+2$, then $\mathcal {C}$ has a reduced code
$[62s,5,32s]=[31\times 2s,5,16\times 2s]$, which is SO.

If $l_{max}= s+1$, it has a reduced code $[62s+1,5,32s]=[31\times
2s+1,5,16\times 2s]$. And this code has a reduced $[60s=15\times
4s,4,16\times 2s]$ SO code, then a $[63s+2,6,32s]$ code has $h\geq
2$ and not LCD by Lemma \ref{li}. Similarly, we can derive that a
$[63s+1,6,32s-1]$ code is not LCD  by Lemma \ref{Bou}.

\end{proof}

\begin{theorem} There is no   $[63s+10,6,32s+4]$ or $[63s+9,6,32s+3]$ LCD code.
\end{theorem}
\begin{proof}
A $[63s+10,6,32s+4]$ code has $\sigma=68$ and then its defining
vector  $L$ satisfies $s-2 \le l_{i}\le s+2$.

If $l_{max}= s+2$, then $\mathcal {C}$ has a reduced code
$[62s+8,5,32s+4]$ violating the Griesmer bound, a contradiction.

When $l_{max}= s+1$, a $[63s+10,6,32s+4]$ code has a reduced code
$[62s+9,5,32s+4]$ code, and $[62s+9,5,32s+4]$ has a reduced
$[60s+8,4,32s+4]$ SO code, thus we can deduce a $[63s+10,6,32s+4]$
code is not LCD. It immediately follows that $[63s+9,6,32s+3]$ code
is not LCD.

\end{proof}

\begin{theorem} There is no   $[63s+14,6,32s+6]$ or $[63s+13,6,32s+5]$ LCD code.
\end{theorem}
\begin{proof}
A $[63s+14,6,32s+6]$ code has $\sigma=70$ and then its defining
vector  $L$ satisfies $s-2 \le l_{i}\le s+2$.

If $l_{max}= s+2$, then $\mathcal {C}$ has a reduced code
$[62s+12,5,32s+6]$ violating the Griesmer bound, a contradiction.

When $l_{max}= s+1$, a $[63s+14,6,32s+6]$ code has a reduced code
$\mathcal{D}=[62s+13,5,32s+6]$ code, which has $h(\mathcal{D})\geq
3$ according to Ref. \cite{liuar}, thus $h(\mathcal {C})$$\geq 2$.
It then follows that $[63s+13,6,32s+5]$ code is not LCD.

\end{proof}

\begin{theorem} There is no  $[63s+17,6,32s+8]$,  $[63s+16,6,32s+7]$, $[63s+17,6,32s+7]$ or $[63s+18,6,32s+8]$
 LCD code.

\end{theorem}

\begin{proof}

A $[63s+17,6,32s+8]$ code has $\sigma=40$ and then its defining
vector  $L$ satisfies $l_{max}=s+1$. Then it has a reduced
$[62s+16,5,32s+8]$ SO code and one can deduce a $[63s+17,6,32s+8]$
code is not LCD by Lemma \ref{li}. It immediately follows that a
$[63s+16,6,32s+7]$ code is not LCD.

Suppose that there is a $[63s+17,6,32s+7]$ LCD code, then there
exists a $[63s+18,6,32s+8]$ LCD  code by Lemma \ref{Bou}. However, a
$[63s+18,6,32s+8]$ code has a reduced  $[63s+17,5,32s+8]$ code with
$h\geq 3$ by Ref. \cite{liuar}. It naturally follows that a
$[63s+18,6,32s+8]$ code doesn't exist and then there is  no
$[63s+17,6,32s+7]$  LCD code.
\end{proof}

\begin{theorem} There is no
$[63s+25,6,32s+12]$ or $[63s+24,6,32s+11]$
 LCD code.
\end{theorem}

\begin{proof}

A $[63s+25,6,32s+12]$ code has $\sigma=44$ and then its defining
vector  $L$ satisfies $s-1\leq l_{i}\leq s+2$.

If $l_{max}= s+2$, then $\mathcal {C}$ has a reduced code
$[62s+23,5,32s+12]$ violating the Griesmer bound, a contradiction.

When $l_{max}= s+1$, it has a reduced  $[62s+24,5,32s+12]$ SO code.
Thus we can deduce a $[63s+25,6,32s+12]$ code  is not LCD, it
follows that a $[63s+24,6,32s+11]$ code is also not LCD.

\end{proof}

\begin{theorem} There is no  $[63s+29,6,32s+14]$ or
$[63s+28,6,32s+13]$  LCD code.
\end{theorem}

\begin{proof}

A $[63s+29,6,32s+14]$  code has $\sigma=46$ and then its defining
vector $L$ satisfies $s-1\leq l_{i}\leq s+2$.

If $l_{max}= s+2$, then $\mathcal {C}$ has a reduced code
$[62s+27,5,32s+14]$ violating the Griesmer bound, a contradiction.

When $l_{max}= s+1$,Then it has a reduced code $[62s+28,5,32s+14]$
code, which is an   $\mathcal{MD}_{s}(5,2)$ code and  has $h=3$ by
Lemma \ref{li2}. Hence a $[63s+29,6,32s+14]$  code has $h\geq 2 $,
thus there is no $[63s+29,6,32s+14]$ LCD code  according to lemma
\ref{li}, which implies a $[63s+28,6,32s+13]$ LCD code does not
exist.
\end{proof}

\begin{theorem} There is no  $[63s+30,6,32s+14]$ or
$[63s+29,6,32s+13]$  LCD code.
\end{theorem}

\begin{proof}
 If $\mathcal {C}$ is a $[63s+30,6,32s+14]$ linear code, then
$\sigma(\mathcal{C})=32\times2+14=78$ and its defining vector $L$
satisfies $s-1\leq l_{i}\leq s+2$.

(1) If $l_{max}=s+2$, then $\mathcal {C}$ has a reduced code
$\mathcal{D}=[62s+28,5,32s+14]$ with $h(\mathcal{D})=3$ by lemma
\ref{li2}, hence $\mathcal{C}$ is not LCD  according to lemma
\ref{li}.

(2) If $l_{max}=s+1$,  a $[63s+30,6,32s+14]$ code has a reduced code
$\mathcal{C}'=[31\times 2s+29,5,16\times2s+14]$ .

For  $\mathcal{C}'$, $\sigma(\mathcal{C}')=16+14=30$ and its
defining vector $L'$ satisfies $2s-1\leq l'_{i}\leq 2s+2$.

(2.1) If $l'_{max}= 2s+2$, then $\mathcal{C}'$ has a reduced code
$[30\times2s+27,5,16\times2s+14]$=$[15\times(4s+1)+12,4,8\times(4s+1)+6]$
violating the Griesmer bound, a contradiction.

(2.2) If $l'_{max}= 2s+1$ and $l'_{min}= 2s$,    $\mathcal {C}$ has
a
 generalized anti-code    $(2,2^{6},\{2\})$ by lemma \ref{li2023lemma}.
Owing to the maximum weight of $2^{6}$ codewords is 2, so it is easy
to know rank$(\mathcal {C}') \le 2$. It then follows that
$h(\mathcal {C}')=5-$rank $(\mathcal {C}'({C}')^{T})\ge 5-2=3$.

(2.3) If $l'_{max}= 2s+1$ and $l'_{min}= 2s-1$, let $m_{1}$, $m_{2}$
and $m_{3}$ be the numbers of equal $2s-1$, $2s$ and $2s+1$,
respectively. Then the following two equations hold:
\begin{eqnarray*}
  m_{1}+m_{2}+m_{3} &=& 31 \\
 (2s-1)\times m_{1}+2s\times m_{2}+(2s+1)\times m_{3} &=& 31\times
 2s+29
\end{eqnarray*}

Solving the above equation system, we can obtain the only one
integer solution: $$ m_{1}=1, m_{2}=0, m_{3}=30.$$ It follows that
the defining vector $L'$ of $\mathcal {C}'$ and the corresponding
anti-vector $L'^{c}$ have the following types:
 $$L': ]](2s-1)_{1}|(2s)_{0}|(2s+1)_{30}]]\quad \hbox{and} \quad L'^{c}: ]](2)_{1}|(1)_{0}|(0)_{30}]].$$
From the form of $L'^{c}$, one can derive that
rank$(G^{c}(G^{c})^{T})=0$ and $\mathcal{C}'$ is SO and
$h(\mathcal{C}')=5$.

Summarizing previous discussions of (1) and (2), we know $h(\mathcal
{C})\ge 2$ by lemma \ref{li} and then a $[63s+61,6,32s+30]$ code
does not exist. Hence, a $[63s+60,6,32s+29]$ LCD code doe not exist.

\end{proof}

\begin{theorem} There is no  $[63s+31,6,32s+15]$, $[63s+32,6,32s+15]$, $[63s+32,6,32s+16]$  or
$[63s+33,6,32s+16]$  LCD code.
\end{theorem}

\begin{proof} A $[63s+32,6,32s+16]$ code is SO  and then a
$[63s+31,6,32s+15]$ LCD  code doesn't exist.

 If $\mathcal {C}$ is a $[63s+33,6,32s+16]$ linear code, then
$\sigma=32+16$ and its defining vector $L$ satisfies $s-1\leq
l_{i}\leq s+2$.

If $l_{max}= s+2$, then $\mathcal {C}$ has a reduced code
$[62s+31,5,32s+16]=[31\times (2s+1),5,16\times (2s+1)]$, which is
SO.

If $l_{max}= s+1$,  a $[62s+32,5,32s+16]=[31\times
(2s+1)+1,5,16\times (2s+1)]$ code has a reduced $[30\times
(2s+1),4,16\times (2s+1)]$ SO code, hence a $[63s+33,6,32s+16]$ has
$h\geq 2$,  $[63s+33,6,32s+16]$ code is not LCD. It follows that a
$[63s+32,6,32s+15]$ code is not  LCD.

\end{proof}

\begin{theorem} There is no   $[63s+41,6,32s+20]$ or $[63s+40,6,32s+19]$
 LCD code.
\end{theorem}

\begin{proof}

A $[63s+41,6,32s+20]$ has $\sigma=52$ and then its defining vector
$L$ with $s-1\leq l_{i}\leq s+2$.

If $l_{max}=s+2$, then $\mathcal {C}$ has a reduced code
$[62s+39=31(2s+1)+8,5,32s+20=16(2s+1)+4]$ violating the Griesmer
bound, a contradiction.

When $l_{max}=s+1$, it has a reduced  $[62s+40,5,32s+20]$ code, and
$[62s+31+9,5,32s+16+4]$ has a reduced  $[60s+30+8,4,32s+16+4]$ SO
code. Thus one can deduce a $[63s+41,6,32s+20]$ code is not LCD
code, it follows that a $[63s+40,6,32s+19]$ is not LCD.
\end{proof}

\begin{theorem}  There is no $[63s+44,6,32s+21]$ or $[63s+45,6,32s+22]$ LCD code.
\end{theorem}

\begin{proof}  If $\mathcal {C}$ is a $[63s+45,6,32s+22]$ linear
code, then $\sigma=32+22$ and its defining vector $L$ satisfies
$s-1\leq l_{i}\leq s+2$.

If $l_{max}= s+2$, then $\mathcal {C}$ has a reduced code
$[62s+43,5,32s+22]$ violating the Griesmer bound, a contradiction.

If $l_{max}= s+1$ and $l_{min}= s$,  there is a unique
$[63s+45,6,32s+22]$ in  \cite{Bou2006}, one can calculate
$h(\mathcal {C})=4$.

If $l_{max}= s+1$ and $l_{min}= s-1$, then  $\mathcal {C}$ has a
reduced code $D=[62s+44,5,32s+22]=[31(2s+1)+13,5, 16(2s+1)+6]$,
which has $h(D)\geq 3$ according to Ref. \cite{liuar}, thus
$h(\mathcal {C})$$\geq 2$. \end{proof}

\begin{theorem}  There is no
 $[63s+47,6,32s+23]$, $[63s+48,6,32s+24]$, $[63s+48,6,32s+23]$ $[63s+49,6,32s+24]$
LCD code.

\end{theorem}

\begin{proof}  If $\mathcal {C}$ is a $[63s+49,6,32s+24]$, then
$\sigma=56$ and its defining vector $L$ satisfies $s-1\leq l_{i}\leq
s+2$.

If $l_{max}= s+2$, then $\mathcal {C}$ has a reduced code
$[62s+47,5,32s+24]$ $=[31(2s+1)+16,5,16(2s+1)+8]$ which is an SO
code, and it easily follows that $h(\mathcal {C}) \ge 4$.

If $l_{max}= s+1$ and  $l_{min}= s$, there are two
$[63s+49,6,32s+24]$ linear codes  in \cite{Bou2006}. Thus we can
further calculate
 $h(\mathcal {C}_{1})$$=6$ and $h(\mathcal {C}_{2})=5$.

If $l_{max}= s+1$ and $l_{min}= s-1$, then  $\mathcal {C}$ has a
reduced code $D=[62s+48,5,32s+24]=[31(2s+1)+17,5, 16(2s+1)+8]$, this
reduced code  has $h(D)\geq 3$ by  Ref. \cite{liuar}, thus
$h(\mathcal {C})$$\geq 2$.

It is easy to know that a $[63s+48,6,32s+24]$ code is SO and then a
$[63s+47,6,32s+23]$  code is not LCD.

Summarizing previous discussions, one can obtain that all of the
 $[63s+47,6,32s+23]$,
$[63s+48,6,32s+24]$, $[63s+48,6,32s+23]$ and $[63s+49,6,32s+24]$
codes are not LCD.
\end{proof}

\begin{lemma}\cite{Li2023}\label{lilem} Let $s\geq 0$, $N_{k}=2^{k}-1$ and $1\leq a\leq 11$.
If $k\geq 4$ for $a=1,3,4,7,8$,  $k\geq 5$ for $a=2,6,10$ and $k\geq
7$ for $a=5,9, 11$, then the corresponding optimal
$[sN_{k}+N_{k}-a,k]$ code is not LCD.
\end{lemma}
From Lemma \ref{lilem}, it is natural to derive the following
conclusion.
\begin{cor} If $s\geq 0$, then
these optimal $[63s+53,6,32s+26]$,  $[63s+55,6,32s+27]$,
$[63s+56,6,32s+28]$, $[63s+57,6,32s+28]$, $[63s+59,6,32s+29]$,
$[63s+60,6,32s+30]$, $[63s+61,6,32s+30]$, $[63s+62,6,32s+31]$ linear
codes are not LCD. Then one can infer that these
$[63s+52,6,32s+25]$, $[63s+56,6,32s+27]$, $[63s+60,6,32s+29]$ codes
are also not LCD   by Lemma \ref{Bou}.
\end{cor}

\section{Constructions of (near) optimal $[n,6]$ LCD codes for $n\geq 51$}\label{sec4}

In last section  we have determined the nonexistence of many classes
of $[n,6,d_{n}]$ LCD codes,  and then one can know it is  optimal
LCD if there exists an $[n,6,d_{n}-1]$ LCD code  or it is near
optimal LCD if there exists an $[n,6,d_{n}-2]$ LCD code. Next we
will construct LCD codes of lengths greater than 50, the most of
which are optimal LCD  while the rest are at least near optimal LCD.
According to the following lemma in Ref. \cite{Li2022}, the optimal
LCD codes can be constructed from these two linear codes $[45,6,22]$
and $[33,6,16]$ and simplex codes of given dimension.

\begin{lemma}\cite{Li2022} \label{con}  Suppose $[n, k_{1}, d_{1}]$ is an LCD code
and there are $[n, k_{1}, d_{1}]\subseteq [n, k, d_{2}] =
\mathcal{C}$ with $h(\mathcal{C}) = k -k_{1}$. If there  is an $[m,
r, d_{3}]$ LCD code, then there is an $[n+m, k, d]$ code with $d
\geq \min\{d_{1}, d_{2} + d_{3}\}$.
\end{lemma}

\subsection{Constructions of   optimal  LCD codes $[n,6,d]$ for $51\leq n\leq
64$}

\begin{theorem}There are optimal LCD codes with parameters
$[51,6,24]$, $[52,6,24]$, $[53,6,25]$, $[54,6,26]$, $[55,6,26]$,
$[56,6,26]$, $[57,6,27]$, $[58,6,28]$, $[59,6,28]$, $[60,6,28]$,
$[61,6,29]$, $[62,6,30]$, $[63,6,30]$ and $[64,6,30]$.
\end{theorem}

\begin{proof} Consider the three simplex matrices $S_{2}$, $S_{4}$ and  $S_{6}$.  Set $$\mathbf{K}_{6,18}=\left(
       \begin{array}{llllll}
S_{4}&\mathbf{0}_{15 \times 3}\\
\mathbf{0}_{3 \times 15}&S_{2}\\
       \end{array}
     \right)$$
and delete the columns of $\mathbf{K}_{6,18}$ from $S_{6}$. Thus,
one can obtain  $$\mathbf{G}_{6,45}=S_{6}\setminus
\mathbf{K}_{6,18}=\left(
       \begin{array}{cccc}
X\\
Y\\
       \end{array}
     \right),$$
which generates a $[45,6,22]$ code with hull dimension 4, where $X$
is a $4\times 45$ matrix and $Y$ is a $2\times 45$ matrix.  $X$
generates a  $[45,4,28]$ code and $Y$ generates a  $[45,2,30]$ LCD
code. Naturally, we have
 the  nested codes $[45,2,30]$ $\subseteq[45,6,22]$. Let  $[m,4,d_{l}(m,4)]$ be an optimal  LCD code for
      $6\leq m\leq 19$ in \cite{Bouyuklieva2021}.

By these  codes $[45,2,30]$, $[45,6,22]$ and  $[m,4,d_{l} ]$, using
Lemma \ref{con}, one can constuct the optimal LCD codes listed  in
the following table.
 \end{proof}

\begin{table}[h!]
\begin{center}

\caption{The optimal LCD $[n,6,d]$ codes  for $51\leq n\leq
64$}\label{tab1} $\begin{tabular}{l|llllllllllllllllllllll}
  \hline
  $n$ & 51 & 52 & 53 & 54 & 55 & 56 & 57 & 58& 59& 60 & 61 & 62 & 63 &64 \\
  \hline
  $d_{l}$ & 24 & 24& 25& 26 & 26& 26& 27 &28 & 28 &28 &29 & 30 &30 & 30 \\
  \hline
\end{tabular}$
 \end{center}
\end{table}

\subsection{Constructions of  optimal LCD codes $[n,6,d]$ for $65\leq n\leq
68$}

\begin{theorem} There are optimal LCD codes with parameters $[65,6,31]$,
$[66,6,32]$, $[67,6,32]$ and $[68,6,32]$.

\end{theorem}

\begin{proof} Consider the  simplex matrix $S_{5}$.
Set $$\mathbf{K}_{6,33}=\left(
       \begin{array}{lllll}
\bf{1_{2}}&\bf{1_{31}}\\
\mathbf{0}_{5 \times 2}&S_{5}\\
       \end{array}
     \right)
$$ and $\mathbf{K}_{6,33}$ generates a $[33,6,16]$
code $\mathcal{C}$, which has a reduced SO code  $[32,5,16]$ and
$$h(\mathbf{K}_{6,33})=6-rank(\mathbf{K} \mathbf{K}^{T})=6-1=5.$$
Obviously there is an LCD code $[33,1,33] \subseteq [33,6,16]$. Let
$[m,5,d_{l}(m,5)]$ be an optimal  LCD code for
      $32\leq m\leq 35$ in \cite{liuar}.

By these  codes $[33,1,33]$, $[33,6,16]$ and  $[m,5,d_{l} ]$, using
Lemma \ref{con}, one can obtain  these  LCD codes $[65,6,31]$,
$[66,6,31]$, $[67,6,32]$ and $[68,6,32]$, respectively.

Further, one can know an optimal $[66,6,32]$ LCD code exists by the
existence of a $[65,6,31]$ LCD code by Lemma \ref{Bou}. Combining
with the above constructions, this lemma holds. For clarity, we list
them in the following table.

\begin{table}[h!]
\begin{center}
\caption{The optimal LCD $[n,6,d]$ codes for $65\leq n\leq 68$
}\label{table3}

$\begin{tabular}{l|llllllllllllllllllllll}
  \hline
  $n$ &  65 & 66 & 67 &68 \\
    \hline
  $d_{l}$ & 31 & 32 &32 & 32 \\
  \hline
\end{tabular}$

 \end{center}
\end{table}
\end{proof}

\subsection{Constructions of (near) optimal LCD codes $[n,6,d]$ for $n=63s+t\ge68$}

When $6\le t \le 68$, the optimal LCD codes with  $[t, 6,
d_{l}(t,6)]$ can be obtained in Refs.
\cite{Grassl,Galvez2018,Harada2019,Fu2019,Araya2020,Bouyuklieva2021,Wang2023}
 and the constructions in the above two subsections.
By the juxtaposition of $s$ simplex codes $[63, 5, 32]$
 and an optimal LCD code $[t, 6, d_{a}(t,6)]$ for $6\le t \le 68$, one can easily obtain all $[n=63s+t, 6, 32s+d_{l}(n,6)]$ (near) optimal
 LCD codes for $n\ge68$. For clarity, we give some examples as
 follows.

\begin{example} When $n=131=63\times1+68$, there exits an optimal
$[131,6,32\times1+32]$ LCD code saturating the Griesmer bound, which
can be constructed by  the juxtaposition of one simplex code $[63,
5, 32]$ and an optimal LCD code $[68,6,32]$. It is obviously also an
optimal code.
\end{example}

\begin{example} When $n=136=63\times 2+10$, there exits a
$[136,6,32\times2+3]$ LCD code, which can be constructed by the
juxtaposition of two simplex codes $[63, 5, 32]$ and an optimal LCD
code $[10,6,3]$. By Theorem 2, we can know there is not
$[63s+10,6,32s+4]$ LCD codes. It naturally follows that the
$[136,6,32\times 2+3]$ LCD code is optimal LCD code, which is also a
near optimal code in Ref. \cite{Grassl}.
\end{example}

\begin{example} When $n=206=63\times3+17$, there exits a
$[206,6,32\times 3+6]$ LCD code, which can be constructed by the
juxtaposition of three simplex codes $[63, 5, 32]$ and an optimal
LCD code $[17,6,6]$. By Theorem 3, we can know there is no
$[63s+17,6,32s+7]$ LCD code. It naturally follows that a
$[206,6,32\times3+6]$ LCD code is optimal LCD, whose minimum
distance is smaller by 2 than the corresponding optimal code in Ref.
\cite{Grassl}.
\end{example}

\begin{example} When $n=273=63\times 4+21$, there exits a
$[273,6,32\times4+8]$ LCD code, which can be constructed by the
juxtaposition of four simplex codes $[63, 5, 32]$ and an optimal LCD
code $[21,6,8]$.  One can infer that it is at least near optimal LCD
and near optimal because the optimal code has the parameter
$[273,6,32\times4+9]$, which can be derived from Ref. \cite{Grassl}.
\end{example}

\section{Conclusion}\label{sec5}
 By the theories of
defining vectors, generalized anti-codes, reduced codes and  nested
code chains, the nonexistence and constructions of LCD codes  have
been studied in last two sections for $n=63s+t\ge 51$. To conclude
the  above results,  $d_{l}(n,6)$ has been exactly determined for $0
\leq t \leq 62$ and $t \notin \{21,22,25,26,33,34,37,38,45,46\}$. We
have also showed that $ d_{l}(n,6)\in$$\{d_{a}(n,6)$
$-1,d_{a}(n,6)\}$ for $t\in \{21,22,26,34,37,38,46\}$ and $
d_{l}(n,6)$ $\in \{d_{a}(n,6)-2,$ $d_{a}(n,6)-1\}$ for
$t\in\{25,33,45\}$.

\begin{table}[h!]
\begin{center}
\caption{ $[n,6,32s+d_t]$ codes with $n=63s+t\ge42$}\label{42}

$\begin{tabular} {lllllllllllllllllll}
 \hline \hline
t&$  0 $  &$   1 $ & $   2 $&$   3 $ &$   4 $&$   5 $&$   6 $&$   7
$&$   8 $&$   9 $&
 $   10 $\\
\hline
 $d_a$&$  0 $  & $  0 $ & $   0 $ & $   0 $ & $  0  $ &
$   0 $ & $   1 $ &  $   2 $&
$   2 $& $   3 $ & $   4 $ \\
 $d_l$&$   -2 $  & $   -2 $ & $   -1 $ & $  0  $ & $  0  $ &
$   0 $ &$   1 $&$    2 $&
$   2 $& $   2 $ & $   3 $ \\
 \hline \hline
  t&&$   11 $&$   12 $&$   13 $& $   14 $
&$   15 $&$   16 $ &  $   17 $&$   18 $
&$   19 $&  $   20 $\\
\hline
 $d_a$&&   $   4 $& $   4 $& $   5 $&  $   6 $&$   6 $&  $   7 $ & $   8 $ & $   8 $ & $   8 $ & $   8 $
  \\
 $d_l$&& $   4 $& $   4 $& $   4 $&  $   5 $& $   6 $& $   6 $ & $   6 $ & $   7 $ &
$   8 $ &$   8 $
  \\
 \hline \hline
t&&$   21 $&$   22 $&$   23 $&$   24 $&
 $   25 $&$   26 $&$   27 $&$   28 $&$   29 $ &$   30 $\\
\hline
 $d_a$&&   $   9 $&
$   10 $& $   10 $& $   11 $ & $   12 $& $   12 $& $   12 $& $ 13
$&$   14 $& $   14 $
  \\
 $d_l$&&   $   8/9 $&
$   9/10 $& $   10 $ & $   10 $& $   10/11 $& $   11/12 $& $   12 $&
$   12 $& $   12 $& $   13 $\\
 \hline \hline

t&&$   31 $&$   32 $ &$   33 $&$   34 $&$   35 $  &$   36 $ & $ 37
$&$   38 $ &$   39 $&
$   40 $\\
\hline
 $d_a$& & $   15 $& $   16 $& $   16 $&$   16 $&$   16 $  & $   16 $ & $   17 $ & $   18 $ & $   18 $ &
$   19 $
  \\
 $d_l$& & $   14 $& $   14 $&
$   14/15 $&$   15/16 $&$   16 $  & $   16 $ & $   16/17 $ & $ 17/18
$ & $   18 $ &
$   18 $\\
 \hline \hline

  n&&$   41 $&$   42 $&$   43 $&$   44 $&
 $   45 $ &$   46 $&$   47 $&$   48 $& $   49 $
&$   50 $\\
\hline
 $d_a$&& $   20 $&  $   20 $&$   20 $& $   21 $ &
$   22 $& $   22 $& $   23 $& $   24 $&  $   24 $& $   24 $
  \\
 $d_l$&&$   19 $&$   20 $&
$   20 $& $   20 $ & $   20/21 $ & $   21/22 $& $   22 $& $   22 $&
$ 23 $& $   24 $
  \\
 \hline \hline
   n&&$   51 $ &  $   52 $&$   53 $
&$   54 $&  $   55 $&$   56 $&$   57 $&$   58 $&$   59 $&
 $   60 $\\
\hline
  $d_a$& &  $   24 $ & $   25 $ & $   26 $ & $   26 $ & $   27 $&   $   28 $&
$   28 $& $   28 $& $   29 $ & $   30 $&
  \\ $d_l$ && $   24 $ & $   24 $ & $   25 $ &
$   26 $ &$   26 $&   $   26 $&
$   27 $& $   28 $ & $   28 $& $   28 $&\\
 \hline \hline
n&&$   61 $&$   62 $\\
\hline
 $d_a$&& $   30 $& $   31 $
  \\
 $d_l$&&$   29 $&
$   30 $\\
 \hline \hline
\end{tabular}$

\qquad  $d_{t }$ denotes $d_a$ or $d_l$ in different lines and the
left slash``/" means ``or". For example, ``14/15" implies the
largest minimum distance of  $[63s+33,6]$ LCD codes is $32s+14$ or
$32s+15$. \qquad \qquad \qquad \qquad \qquad \qquad \qquad \qquad
\qquad \qquad \qquad \qquad \qquad \qquad \qquad \qquad \qquad
\qquad \qquad
 \end{center}
\end{table}

 From Ref. \cite{Grassl}, it
is easy to know all optimal [n,6] linear codes can achieve the
Griesmer bound for $42 \leq n \leq 256$. For $n>256$, the length $n$
can be denoted as
 $n=63s+t$,  where $s\geq 3$ and $63\leq t\leq 125$ are integers.
 By the juxtaposition of $s$ simplex codes $[63,5,16]$
 and an optimal linear code $[t,6, d_{a}(t,5)]$, one can  obtain  all $[n, 6, d_{a}(n,6)]$ optimal
 linear codes with $d_{a}(n, 6)$ achieving the Griesmer bound for $
 n>256$. That is to say any $d_{a}(n, 6)$ can meet the Griesmer
 bound for all lengths $n\geq 42$.

Let $d_{a}$ and $d_{l}$ be defined as the previous sections, we list
optimal linear codes and (near) optimal LCD codes in Table \ref{42}
for $n\geq 42$. Among them, $d_l$ of (near) optimal LCD codes of
lengths $n\geq 51$  are obtained in this paper.

\begin{remark}

From Table \ref{42},  these codes can be clearly divided into four
groups:

 Let $n=63s+t\geq 42$, the following items hold:

(i) For $t$=3, 4, 5, 6, 7, 8, 11, 12, 15,  19, 20, 23, 27, 35, 36,
39, 42, 43,  50, 51, 54, 58,  the corresponding $[n,6,d_{l }]$
optimal LCD codes are also optimal linear codes.

(ii) For $t$=2, 9, 10, 13, 14, 16, 18, 21, 22,  24, 25, 26, 28, 30,
31, 34, 37, 38, 40, 41, 42,  44,  46, 47, 49, 52, 53, 55, 57, 59,
61, 62, the corresponding $[n,6,d_{l }]$  LCD codes are near optimal
linear codes, i.e. $d_{l }=d_{a }-1$.  If $t\neq$  21, 22, 26,  34,
37, 38, 46,  the  $[n,6,d_{l }]$ codes are  optimal LCD codes.

 (iii) For $t$=0, 1, 17, 25, 29, 32, 33, 45, 48, 56, 60,  the
corresponding $[n,6,d_{l }]$  LCD codes are not near optimal linear
codes and $d_{l }=d_{a }-2$. If $t\neq$  25, 33, 45,  the $[n,6,d_{l
}]$ codes are  optimal  LCD codes.

 (iv) For each $t$= 21, 22, 25, 26, 33,  34, 37, 38, 45, 46,
  the corresponding $[n,6,d_{l }]$ is at least  near optimal LCD
codes and  the optimal LCD codes  still can not determined.  We
sincerely look forward to the completion of the follow-up work. When
the dimensions of LCD codes are higher, the approach in the paper
can be also employed, one can further study optimal LCD codes with
higher dimensions.  It is hoped that  more results about optimal LCD
codes will be obtained by more scholars in future study.

\end{remark}



\bibliographystyle{model1a-num-names}



\section*{Acknowledgements}

This work is supported by Natural Science Foundation of Shaanxi
Province under Grant Nos.2024JC-YBMS-055, 2023-JC-YB-003,
2023-JC-QN-0033, National Natural Science Foundation of China under
Grant No.U21A20428.

\end{document}